\documentclass[12pt, twoside]{article}

\usepackage{amsmath,amsthm,amssymb}

\usepackage{times}

\usepackage{enumerate}

 \usepackage{amsmath}
\usepackage{amssymb}
\usepackage{amsthm}
\usepackage{graphicx}
\usepackage{amscd}
\usepackage{graphicx}

\usepackage{amssymb}
\usepackage{amsmath}
\usepackage{latexsym}
\usepackage[numbers,sort&compress]{natbib}
\usepackage{verbatim}
\usepackage{physics}
\usepackage{dsfont}	
\usepackage{microtype}  
\usepackage{hyperref}
\usepackage{tikz}
\usepackage{booktabs}   

\usepackage{color}

\newtheorem{theorem}{Theorem}
\newtheorem{lemma}[theorem]{Lemma}

\newtheorem{example}[theorem]{Example}

\begin{document}
 
\title{The geometry of Hermitian self-orthogonal codes}
\author{Simeon Ball and Ricard Vilar\thanks{17 August 2021}}
 
\date{}

\maketitle

\begin{abstract}
We prove that  if $n >k^2$ then a $k$-dimensional linear code of length $n$ over ${\mathbb F}_{q^2}$ has a truncation which is linearly equivalent to a Hermitian self-orthogonal linear code. In the contrary case we prove that truncations of linear codes to codes equivalent to Hermitian self-orthogonal linear codes occur when the columns of a generator matrix of the code do not impose independent conditions on the space of Hermitian forms. In the case that there are more than $n$ common zeros to the set of Hermitian forms which are zero on the columns of a generator matrix of the code, the additional zeros give the extension of the code to a code that has a truncation which is equivalent to a Hermitian self-orthogonal code.
\end{abstract}

\section{Introduction}

The main motivation to study Hermitian self-orthogonal codes is their application to quantum error-correcting codes. The most prevalent and applicative quantum codes are qubit codes, in which the quantum state is encoded on $n$ quantum particles with two-states. In this case, the quantum code is a subspace of $({\mathbb C}^2)^{\otimes n}$. More generally, a quantum code is a subspace of $({\mathbb C}^q)^{\otimes n}$. The parameter $q$ is called the {\em local dimension} and corresponds to the number of states each quantum particle of the system has. A qubit is then referred to as a quqit.

A quantum code with minimum distance $d$ is able to detect errors, which act non-trivially on the code space, on up to $d-1$ of the quqits and correct errors on up to $\frac{1}{2}(d-1)$ of the quqits. If the code encodes $k$ logical quqits onto $n$ quqits then we say the quantum code is an $[\![ n,k,d ]\!] _q$ code. It has dimension $q^k$.

Suppose that $q=p^h$ is a prime power and let ${\mathbb F}_q$ denote the finite field with $q$ elements. A linear code $C$ of length $n$ over ${\mathbb F}_q$ is a subspace of ${\mathbb F}_q^n$. If the minimum weight of a non-zero element of $C$ is $d$ then the minimum (Hamming) distance between any two elements of $C$ is $d$ and we say that $C$ is a $[n,k,d]_q$ code, where $k$ is the dimension of the subspace $C$. If we do not wish to specify the minimum distance then we say that $C$ is a $[n,k]_q$ code.

A canonical Hermitian form on ${\mathbb F}_{q^2}^n$ is given by
$$
(u,v)_h=\sum_{i=1}^n u_iv_i^q.
$$
If $C$ is a linear code over ${\mathbb F}_{q^2}$ then its {\em Hermitian dual} is defined as
$$
C^{\perp_h}=\{ v \in {\mathbb F}_{q^2}^n \ | \ (u,v)_h=0, \ \mathrm{for} \ \mathrm{all} \ u \in C \}.
$$
The standard dual of $C$ will be denoted by $C^{\perp}$. Observe that $v \in C^{\perp}$ if and only if $v^q \in C^{\perp_h}$, so both of the dual codes have the same weight distribution.

One very common construction of quantum stabiliser codes relies on the following theorem from Ketkar et al. \cite[Corollary 19]{KKKS2006}. It is a generalisation from the qubit case of a construction introduced by Calderbank et al. \cite[Theorem 2]{CRSS1998}.

\begin{theorem} \label{sigmaortog}
If there is a $[n,k]_{q^2}$ linear code $C$ such that $C \leqslant C^{\perp_h}$ then there exists an $ [\![ n,n-2k,d]\!] _q$ quantum code, where $d$ is the minimum weight of the elements of $C^{\perp_h} \setminus C$ if $k \neq \frac{1}{2}n$ and $d$  is the minimum weight of the non-zero elements of $C^{\perp_h}=C$ if $k=\frac{1}{2}n$.
\end{theorem}

If $C \leqslant C^{\perp_h}$ then we say the linear code $C$ is {\em Hermitian self-orthogonal}. Theorem~\ref{sigmaortog} is our motivation to study Hermitian self-orthogonal codes. We can scale the $i$-th coordinate of all the elements of $C$ by a non-zero scalar $v_i$, without altering the parameters of the code. Such a scaling, together with a reordering of the coordinates, gives a code which is said to be {\em linearly equivalent} to $C$. 

Thus, a linear code $D$ is {\em linearly equivalent} to a linear code $C$ over ${\mathbb F}_q$ if, after a suitable re-ordering of the coordinates, there exist non-zero $\theta_i \in {\mathbb F}_q$ such that 
$$
D=\{ (\theta_1u_1,\ldots,\theta_n u_n) \ | \ (u_1,\ldots,u_n) \in C\}.
$$
A {\em truncation} of a code is a code obtained from $C$ by deletion of coordinates. Observe that a truncation can reduce the dimension of the code but the dual minimum distance can only increase. 

We will be interested in the following question: Given a linear $[n,k,d]_q$ code $C$, what truncations does $C$ have which are linearly equivalent to Hermitian self-orthogonal codes? 

In the special case that $C$ is a $k$-dimensional Reed-Solomon code, the above question was answered by the authors in \cite{BV2021}. The code $C$ has a truncation of length $m \leqslant q^2$ which is linearly equivalent to a Hermitian self-orthogonal code if and only if there is a polynomial $g(X) \in {\mathbb F}_{q^2}[X]$ of degree at most $(q-k)q-1$ with the property that
$g(x)+g(x)^q$, when evaluated at the elements $x \in {\mathbb F}_{q^2}$, has precisely $q^2+1-m$ zeros.

\section{Hermitian self-orthogonal codes} \label{puncturecodesection}

Let $C$ be a linear code of length $n$ over ${\mathbb F}_{q^2}$. We have that $C$ is linearly equivalent to a Hermitian self-orthogonal code if and only if there are non-zero $\theta_i \in {\mathbb F}_{q^2}$ such that
$$
\sum_{i=1}^n \theta_{i}^{q+1}u_{i}v_{i}^q=0,
$$
for all $u,v \in C$. Note that $\theta_{i}^{q+1}$ is an non-zero element of ${\mathbb F}_q$, so equivalently 
$C$ is linearly equivalent to a Hermitian self-orthogonal code if and only if there are non-zero $\lambda_i \in {\mathbb F}_{q}$ such that
$$
\sum_{i=1}^n \lambda_{i} u_{i}v_{i}^q=0.
$$

For any linear code $C$ over ${\mathbb F}_{q^2}$ of length $n$, Rains \cite{Rains1999} defined the {\em puncture code} $P(C)$ to be
$$
P(C)=\{\lambda=(\lambda_{1},\ldots,\lambda_n) \in {\mathbb F}_q^n \ | \  \sum_{{i}=1}^n \lambda_{i} u_{i}v_{i}^q=0, \ \mathrm{for} \ \mathrm{all} \ u,v \in C \}.
$$
Then, clearly we have the following theorem.

\begin{theorem} \label{puncturecodethm}
There is a truncation of a linear code $C$ over ${\mathbb F}_{q^2}$ of length $n$ to a linear over ${\mathbb F}_{q^2}$ of length $r \leqslant n$ which is linearly Hermitian self-orthogonal code if and only if there is an element of $P(C)$ of weight $r$.
\end{theorem}

Thus, as emphasised in \cite{GR2015}, the puncture code is an extremely useful tool in constructing Hermitian self-orthogonal codes. Observe that, the minimum distance of any quantum code, given by an element in the puncture code, will have minimum distance at least the minimum distance of $C^{\perp}$, since any element in the dual of the shortened code will be an element of $C^{\perp}$ if we replace the deleted coordinates with zeros. 

Given a linear code $C$ over ${\mathbb F}_{q^2}$ it is not obvious how one can efficiently compute the puncture code. Let $\mathrm{G}=(g_{i\ell})$ be a generator matrix for $C$, i.e. a $k \times n$ matrix whose row-span is $C$. A straightforward approach would be to construct a ${k+1 \choose 2} \times n$ matrix $\mathrm{T}(\mathrm{G})=(t_{ij,\ell})$ over ${\mathbb F}_{q^2}$, where for $\{i,j\} \subseteq \{1,\ldots,k\}$ we define
\begin{equation} \label{defineT}
t_{ij,\ell}=\left\{ \begin{array}{ll}
 g_{i\ell}g_{j\ell}^q & i<j, \\
 g_{i\ell}^{q+1} & i=j.  
 \end{array} \right.
 \end{equation}
 \begin{lemma} \label{rKerT}
 The puncture code $P(C)$ is the intersection of the right-kernel of $\mathrm{T}(\mathrm{G})$ with ${\mathbb F}_q^n$.
 \end{lemma}
 
 \begin{proof}
For any $u,v \in C$, 
$$
u_{\ell}=\sum_{i=1}^k a_ig_{i\ell} \ \ \mathrm{and}\  \ v_{\ell}=\sum_{j=1}^k b_j g_{j\ell}
$$ 
for some $(a_1,\ldots,a_k) \in {\mathbb F}_q^k$ and $(b_1,\ldots,b_k) \in {\mathbb F}_q^k$.

Since
$$
\sum_{{\ell}=1}^n \lambda_{\ell} u_{\ell}v_{\ell}^q =\sum_{i=1}^k \sum_{j=1}^k a_i b_j ^q \sum_{{\ell}=1}^n \lambda_{\ell} g_{i\ell} g_{j\ell}^q,
$$
we have that $\lambda=(\lambda_1,\ldots,\lambda_n)$ is in the right-kernel of $\mathrm{T}(\mathrm{G})$ if and only if
$$
\sum_{{\ell}=1}^n \lambda_{\ell} u_{\ell}v_{\ell}^q=0,
$$
for all $u,v \in C$.
\end{proof}

Thus, the puncture code $P(C)$ can then be found by extracting the elements in the right-kernel of $\mathrm{T}(\mathrm{G})$ all of whose coordinates are elements of ${\mathbb F}_q$. However, this quickly becomes unfeasible computationally for larger parameters. 

Our first aim, which we will deal with now, is to construct a parity check matrix for $P(C)$, i.e. a matrix {\em over ${\mathbb F}_q$} whose right-kernel is $P(C)$. This allows one to determine, given a linear code $C$ over ${\mathbb F}_{q^2}$, all truncations of $C$ which are linearly equivalent to a Hermitian self-orthogonal code, provided that the dimension of $P(C)$ is not too large.

Let $e \in {\mathbb F}_{q^2} \setminus {\mathbb F}_q$.

Let $\mathrm{M}(\mathrm{G})=(m_{ij,\ell})$ be a $k^2 \times n$ matrix where, for $i,j \in \{1,\ldots,k\}$, we define  
$$
m_{ij,\ell}=\left\{ \begin{array}{ll} 
e g_{i\ell}g_{j\ell}^q+e^q g_{i\ell}^qg_{j\ell} & i<j \\
g_{i\ell}g_{j\ell}^q+ g_{i\ell}^qg_{j\ell} & i>j \\
g_{i\ell}^{q+1} & i=j
\end{array}
\right. .
$$

\begin{theorem} \label{rkerM}
The matrix $\mathrm{M}(\mathrm{G})$ is a parity check matrix for $P(C)$. i.e. $\mathrm{M}(\mathrm{G})$ is defined over ${\mathbb F}_{q}$ and its right-kernel is $P(C)$.
\end{theorem}

\begin{proof}
Observe first that all the entries in the matrix $\mathrm{M}(\mathrm{G})$ are in ${\mathbb F}_{q}$.

Suppose that $\lambda=(\lambda_1,\ldots,\lambda_n)$ is in the right-kernel of $\mathrm{M}(\mathrm{G})$. Hence, for all $i,j \in \{1,\ldots,k\}$  with $i<j$,
$$
\sum_{\ell=1}^n \lambda_{\ell}( e g_{i\ell}g_{j\ell}^q+e^q g_{i\ell}^qg_{j\ell})=0
$$
and 
$$
\sum_{\ell=1}^n \lambda_{\ell}  (g_{j\ell}g_{i\ell}^q+ g_{j\ell}^qg_{i\ell})=0.
$$
Multiplying the latter equation by $e^q$ and subtracting the former implies
$$
(e^q-e)\sum_{\ell=1}^n \lambda_{\ell} g_{i\ell}  g_{j\ell}^q=0.
$$
Since $\lambda=(\lambda_1,\ldots,\lambda_n)$ is in the right-kernel of $\mathrm{M}(\mathrm{G})$ we also have that
$$
\sum_{\ell=1}^n \lambda_{\ell}  g_{i\ell}^{q+1}=0.
$$
Hence, $\lambda$ is in the right-kernel of $\mathrm{T}(\mathrm{G})$. 

Since it is also in ${\mathbb F}_q^n$, by Lemma~\ref{rKerT}, $\lambda \in P(C)$.

Suppose that $\lambda=(\lambda_1,\ldots,\lambda_n) \in P(C)$. Then, for all $i,j \in \{1,\ldots,k\}$,
$$
\sum_{\ell=1}^n \lambda_{\ell} g_{i\ell}  g_{j\ell}^q=0.
$$
This implies that $\lambda$ is in the right-kernel of $\mathrm{M}(\mathrm{G})$.
\end{proof}

 \begin{example}
Theorem~\ref{rkerM} can allow us to efficiently calculate the puncture code of a linear code. Then for each codeword of weight $r$ in the puncture code, by Theorem~\ref{puncturecodethm}, we can construct a quantum error correcting code of length $r$. For example, let $C$ be the linear $[43,7]_4$ code, which is dual to the cyclic linear $[43,36,5]_4$ code, constructed from the divisor of $x^{43}-1$,
$$
x^7 + e x^5 + x^4 + x^3 + e^2x^2 + 1,
$$
where $e$ is a primitive element of ${\mathbb F}_4$.

By Theorem~\ref{rkerM}, we can calculate the puncture code from the $49 \times 43$ matrix $\mathrm{M}$ over ${\mathbb F}_2$, which turns out to have rank $29$. The puncture code $P(C)$ has weights $14+2j$ for all $j \in \{0,1,2,3,4,5,6,7,8\}$.



The truncations to codes of length $14$ give a $[14,7,6]_4$ code which is equal to its Hermitian dual. By Theorem~\ref{sigmaortog}, this implies the existence of a $[\![14,0,6]\!]_2$ quantum code.

The truncations to codes of length $18+2j$ give a $[18+2j,7]_4$ code with dual minimum distance $5$, which by Theorem~\ref{sigmaortog} implies the existence of a $[\![18+2j,4+2j,5]\!]_2$ quantum code, for all $j \in \{0,1,2,3,4,5,6\}$.

These codes equal the best known qubit error-correcting codes, according to Grassl \cite{codetables}.
 
\end{example}

\begin{example} 
Consider the dual $C$ to the cyclic linear $[51,42,6]_4$ code, constructed from the divisor of $x^{51}-1$,
$$
x^9 + e^2x^8 + ex^6 + x^5 + e^2x^4 + e^2x^2 + e^2x + 1.
$$
The dimension of the puncture code $P(C)$ is $10$. The puncture code $P(C)$ has codewords of weight $18+2j$, for all $j \in \{0,2,3,4,6,7,8\}$, which implies that it truncates to codes equivalent to Hermitian self-orthogonal codes of length $18+2j$. One can check these are $[18+2j,9]_4$ codes with dual minimum distance $6$. By Theorem~\ref{sigmaortog}, this implies the existence of a $[\![18+2j,2j,6]\!]_2$ quantum code, for all $j \in \{0,2,3,4,6,7,8\}$. 
\end{example}
 
\begin{example} 
Consider $C$ the $[15,5]_9$ code with generator matrix 
$$
\mathrm{G}=
\left(
\begin{array}{ccccccccccccccc}
        1& 1& 1& 1& 1& 0& 0& 0& 0& 
      0& 1& 0& 0& 0& 0 \\ 
   0& 0& 0& 0& 0& 1& 1& 1& 1& 
      1& 0& 1& 0& 0& 0 \\ 
   e^7& e^6& e^5& e^4& 1& e& e^3& e^5& 
      e^4& 1& 0& 0& 1& 0& 0 \\ 
   e^3& e& e^4& e^5& 1& e^6& e^7& e^4& 
      e^5& 1& 0& 0& 0& 1& 0 \\ 
   e^6& e^7& e^5& e^2& e^4& e^2& e^6& 
      e^7& e^3& 1& 0& 0& 0& 0& 1 
      \end{array}
      \right)
$$
The dual code $C^{\perp}$ is a linear $[15,10,5]_9$ code. The dimension of the puncture code $P(C)$ is $2$ and has codewords of weight $9,12$ and $15$. This implies that it truncates to codes equivalent to Hermitian self-orthogonal codes of length $9$, 12 and 15 and one can check that these codes are a $[9,4]_9$, a $[12,5]_9$ and a $[15,5]_9$ codes all with dual minimum distance $5$. By Theorem~\ref{sigmaortog}, this implies the existence of a $[\![9,1,5]\!]_3$, a $[\![12,2,5]\!]_3$ and a $[\![15,5,5]\!]_3$ code. The former of these attains the quantum Singleton bound, proved by Rains in \cite{Rains1999}, which states that
$$
k\leqslant n-2(d-1).
$$
It was proven in \cite{BV2021} that a $[9,4,6]_9$ MDS code does not come from a truncation of a generalised Reed-Solomon code. The only $[9,4,6]_9$ code which is not the truncation of a generalised Reed-Solomon code is the projection of Glynn's $[10,5,6]_9$ MDS code, see \cite{Glynn1986}.
\end{example}

 
\section{The geometry of Hermitian self-orthogonal codes}

Let $\mathrm{PG}(k-1,q)$ denote the $(k-1)$-dimensional projective space over ${\mathbb F}_q$. 

A Hermitian form is given by
$$
H(X)=\sum_{1 \leqslant i<j \leqslant k} (h_{ij}X_iX_j^q+ h_{ij}^qX_i^qX_j)+\sum_{i=1}^k h_{ii}^{q+1} X_i^{q+1}.
$$
for some $h_{ij} \in {\mathbb F}_{q^2}$.

The set of Hermitian forms 
is a $k^2$-dimensional vector space over ${\mathbb F}_q$.


Let $\mathrm{G}=(g_{i\ell})$ be a $k \times n$ generator matrix for a linear code $C$ whose dual minimum distance is at least three. Let $\mathcal{X}$ be the set of columns of $\mathrm{G}$ considered as points of $\mathrm{PG}(k-1,q)$. Observe that the condition that the dual code of $C$ has minimum distance at least three ensures that $\mathcal X$ is a set (and not a multi-set). Such a code is often called a {\em projective code}. Observe that the set $\mathcal{X}$ is the same for all codes linearly equivalent to $C$. Let $\mathrm{HF}(\mathcal X)$ be the subspace of Hermitian forms that are zero on $\mathcal X$.

\begin{lemma} \label{LkerM}
The dimension of the left kernel of the matrix $\mathrm{M}(\mathrm{G})$ is equal to $\dim\mathrm{HF}(\mathcal X)$.
\end{lemma}

\begin{proof}
Let $x \in \mathcal X$ and consider a vector $v$ in the left kernel of $\mathrm{M}(\mathrm{G})$. 

Observe that the coordinates of $v$ are indexed by $i,j \in \{1,\ldots,k\}$. 

Since $x$ is a column of $\mathrm{G}$,
$$
\sum_{i,j=0}^k v_{ij} (e x_{i}x_{j}^q+e^q x_{i}^q x_{j}) +
v_{ji}( x_{i}x_{j}^q+ x_{i}^qx_{j} )+
\sum_{i=1}^k v_{ii}x_{i}^{q+1}=0.
$$
Thus, defining 
$$
h_{ij}=ev_{ij}+v_{ji} \ \ \mathrm{and} \ \ h_{ii}^{q+1}=v_{ii},
$$
we have that 
$$
H(x)=0.
$$
Letting $v$ run over a basis for the left kernel of $\mathrm{M}(\mathrm{G})$, we obtain a set of linearly independent Hermitian forms. Indeed, let $B$ be a basis for the left kernel of $\mathrm{M}(\mathrm{G})$. Suppose there are $\lambda_v \in {\mathbb F}_q$, for $v \in B$, not all zero, such that, for all $i,j \in \{1,\ldots,k\}$,
$$
\sum_{v \in B} \lambda_v(ev_{ij}+v_{ji})=0, \ \  \sum_{v \in B}  \lambda_v v_{ii}=0. 
$$
Since $e \in {\mathbb F}_{q^2} \setminus {\mathbb F}_q$, this implies
$$
\sum_{v \in B} \lambda_v v_{ij}=0,
$$
for all $i,j \in \{1,\ldots,k\}$, contradicting the fact that $B$ is a basis.
 
Vice-versa, if $H(x)=0$ for some Hermitian form $H$, then we obtain $v_{ij}$ by solving   
$$
h_{ij}=ev_{ij}+v_{ji} \ \ \mathrm{and} \ \ h_{ij}^q=e^qv_{ij}+v_{ji},
$$
where $v_{ij},v_{ji} \in {\mathbb F}_q$, and $v_{ii}=h_{ii}^{q+1}$.
Letting $H$ run over a basis for $\mathrm{HF}(\mathcal X)$, we obtain a set of linearly independent vectors in the left kernel of the matrix $\mathrm{M}(\mathrm{G})$.
\end{proof}

The previous lemma allows us to calculate the dimension of the puncture code in terms of the dimension of the space of Hermitian forms which are zero on $\mathcal X$. In the following $\mathcal X$ is obtained, as before, as the set of columns of a generator matrix for $C$, viewed as points of $\mathrm{PG}(k-1,q)$. Note, that the statement that $\mathcal X$ imposes $r$ conditions on the space of Hermitian forms is to say that the co-dimension of $\mathrm{HF}(\mathcal X)$ is $r$.

\begin{theorem} \label{dimpc}
The set $\mathcal X$ imposes $n-\dim P(C)$ conditions on the space of Hermitian forms and
$$
\dim P(C)=n-k^2+\dim \mathrm{HF}(\mathcal X).
$$ 
\end{theorem}

\begin{proof}
By Lemma~\ref{LkerM},
$$
\dim \mathrm{HF}(\mathcal X)=\dim \mathrm{left} \, \mathrm{kernel}\  \mathrm{M}(\mathrm{G})=k^2-\mathrm{rank}\; \mathrm{M}(\mathrm{G}).
$$
By Theorem~\ref{rkerM},
$$
n-\mathrm{rank}\; \mathrm{M}(\mathrm{G}))=\dim P(C),
$$
which proves the second statement. For the first statement, observe that $\dim \mathrm{HF}(\mathcal X)= k^2-r$, where $r$ is the number of conditions imposed by $\mathcal X$ on the space of Hermitian forms.
\end{proof}

Note that in the following statements the truncation may be the code itself.

\begin{theorem} \label{nconds}
The set of points $\mathcal X$ imposes $|\mathcal X|$ conditions on the space of Hermitian forms if and only if no truncation of $C$ is equivalent to a Hermitian self-orthogonal code.
\end{theorem}

\begin{proof}
Theorem~\ref{dimpc} implies that the set of points $\mathcal X$ imposes $n$ conditions on the space of Hermitian forms if and only if $\dim P(C)=0$ which, by Theorem~\ref{puncturecodethm}, is if and only if no truncation of $C$ is equivalent to a Hermitian self-orthogonal code.
\end{proof}

Thus, from Theorem~\ref{nconds}, we deduce that to find codes contained in their Hermitian dual it is necessary and sufficient to find a set of points $\mathcal X$ which does not impose $|\mathcal X|$ conditions on the space of Hermitian forms.

\begin{theorem}  \label{notnconds}
The set of points $\mathcal X$ imposes less than $|\mathcal X|$ conditions on the space of Hermitian forms if and only if some truncation of $C$ is linearly equivalent to a Hermitian self-orthogonal code.
\end{theorem}

Theorem~\ref{notnconds} has some immediate consequences. 
\begin{theorem}
A linear $[n,k]_{q^2}$ code for which $n>k^2$ has a truncation which is linearly equivalent to Hermitian self-orthogonal code.
\end{theorem}

\begin{proof}
Since $n$ is larger than the dimension of the space of Hermitian forms, $\mathcal X$ cannot impose $n$ conditions on the space of Hermitian forms. Hence, Theorem~\ref{notnconds} implies the statement.
\end{proof}

\begin{example}
Let $e$ be a primitive element of ${\mathbb F}_9$, where $e^2=e+1$. Let $D$ be the cyclic linear $[73,66,6]_9$ code, constructed from the divisor of $x^{73}-1$,
$$
x^7 + e x^6 + e^6x^5 + e^3x^4 + e^7x^3 + e^2x^2 +e^5x+ 2.
$$
Let $C$ be the $[60,7]$ code obtained from $D^{\perp}$ be deleting coordinates 61 to 73. The dimension of the puncture code $P(C)$ is $11$. The puncture code $P(C)$ has codewords of weight $\{26,27,\ldots,55\}$ which implies the existence of a $[\![n,n-14, 6]\!]_3$ quantum codes, for all $n \in \{26,27,\ldots,55\}$.
\end{example}

The previous theorem and following theorem are the main results of this paper.
 
\begin{theorem} \label{noshortherm}
A linear $[n,k]_{q^2}$ code $C$ of length $n$ over ${\mathbb F}_{q^2}$ which has no truncations which are linearly equivalent to a Hermitian self-orthogonal code can be extended to $C'$, a $[n+1,k]_{q^2}$ code which does have a truncation to a code which is linearly equivalent to a Hermitian self-orthogonal code, if and only if $\mathcal X$ imposes $n$ conditions on the space of Hermitian forms and the set of common zeros of $\mathrm{HF}(\mathcal X)$ is larger than $|\mathcal X|$. 
\end{theorem}

\begin{proof}
($\Rightarrow$) Let $\mathcal X'$ be the set of columns of a generator matrix for $C'$ obtained by extending the matrix $\mathrm{G}$. By Theorem~\ref{dimpc}, both $\mathcal X$ and $\mathcal X'$ impose $n$ conditions on the space of Hermitian forms. Hence,  
$$
\mathrm{HF}(\mathcal X)=\mathrm{HF}(\mathcal X')
$$
 which implies that the set of common zeros of $\mathrm{HF}(\mathcal X)$ contains $\mathcal X'$.
 
($\Leftarrow$) Let $\mathcal X'=\mathcal X \cup \{x\}$ be a subset of the set of common zeros of $\mathrm{HF}(\mathcal X)$. Let $C'$ be the code with generator matrix whose columns are the elements of $\mathcal X'$. Then $\mathcal X'$ imposes $n$ conditions on the space of Hermitian forms, so Theorem~\ref{dimpc} implies that $\dim P(C')=1$. Thus, $C'$ extends $C$ to  a $[n+1,k]_{q^2}$ code which has a truncation to a code which is linearly equivalent to a Hermitian self-orthogonal code.
 \end{proof}

Theorem~\ref{noshortherm} indicates that to extend a linear code $C$ to a Hermitian self-orthogonal code, we should calculate the set of common zeros of the Hermitian forms which are zero on the columns of a generator matrix for $C$. 

\begin{example} The $[13,7]_4$ code generated by the matrix
$$
\mathrm{G}=
\left(
\begin{array}{ccccccccccccc}
  1& 0& 0& 0& 0& 0& 0& 1& e& 
      0& e^2& e& e \\ 
   0& 1& 0& 0& 0& 0& 0& 0& e& 
      e& e& 0& e^2 \\ 
   0& 0& 1& 0& 0& 0& 0& 1& 1& 
      e& e^2& 1& 0 \\ 
   0& 0& 0& 1& 0& 0& 0& e& 1& 
      0& e& 0& e^2 \\ 
   0& 0& 0& 0& 1& 0& 0& 0& e^2& 
      e^2& e& e& 0 \\ 
   0& 0& 0& 0& 0& 1& 0& e^2& 
      e^2& e& 1& e^2& e \\ 
   0& 0& 0& 0& 0& 0& 1& 1& 1& 
      1& e^2& e& e^2 
\end{array}
\right)
      $$
has dual minimum distance $6$. As before, let $\mathcal X$ be the $13$ points which are the columns of the matrix $\mathrm{G}$. The dimension of $\mathrm{HF}(\mathcal X)$ is 36, so $\mathcal X$ imposes $13$ conditions on the space of Hermitian forms. Theorem~\ref{dimpc} implies that $\dim P(C)=0$, so $C$ has no truncations which are linearly equivalent to Hermitian self-orthogonal codes. However, there are $14$ points which are common zeros of the zeros of $\mathrm{HF}(\mathcal X)$, the points of $\mathcal X$ and the point
$$
(0,e,0,1,e,1,1).
$$
Thus, Theorem~\ref{noshortherm} implies that the $[14,7]_4$ code, with generator matrix
$$
\left(
\begin{array}{cccccccccccccc}
  1& 0& 0& 0& 0& 0& 0& 1& e& 
      0& e^2& e& e & 0 \\ 
   0& 1& 0& 0& 0& 0& 0& 0& e& 
      e& e& 0& e^2 & e\\ 
   0& 0& 1& 0& 0& 0& 0& 1& 1& 
      e& e^2& 1& 0 & 0\\ 
   0& 0& 0& 1& 0& 0& 0& e& 1& 
      0& e& 0& e^2 & 1\\ 
   0& 0& 0& 0& 1& 0& 0& 0& e^2& 
      e^2& e& e& 0 & e\\ 
   0& 0& 0& 0& 0& 1& 0& e^2& 
      e^2& e& 1& e^2& e & 1\\ 
   0& 0& 0& 0& 0& 0& 1& 1& 1& 
      1& e^2& e& e^2 & 1\\
\end{array}
\right)
      $$
      has a truncation which is Hermitian self-orthogonal. Indeed, one can check that the code itself is Hermitian self-orthogonal.  Thus, from this code we can construct, by Theorem~\ref{sigmaortog}, a $[\![14,0,6]\!]_2$ code. 
\end{example}

\section{Conclusions and further work}

In conclusion, we give a summary of the main results.

Suppose that $C^{\perp}$ is a $[n,n-k,d]_{q^2}$, where $d \geqslant 3$.

If $n >k^2$ then we have shown that there are truncations of $C$ which are linearly equivalent to Hermitian self-orthogonal codes. 

If $n\leqslant k^2$ and $\dim P(C)>0$ then there are truncations of $C$ which are linearly equivalent to Hermitian self-orthogonal codes. 

If $n\leqslant k^2$ and $\dim P(C)=0$ and there are points which are not in $\mathcal X$ but are zeros of the forms in $\mathrm{HF}(\mathcal X)$ then we can extend $C$ to a $[n+1,k]_{q^2}$ which does have truncations which are linearly equivalent to Hermitian self-orthogonal codes. 

Finally, if $n\leqslant k^2$ and $\dim P(C)=0$ and there are no points which are zeros of the forms in $\mathrm{HF}(\mathcal X)$ but which are not in $\mathcal X$ then $C$ has no extension to a $[n+1,k]_{q^2}$ which has truncations that are linearly equivalent to Hermitian self-orthogonal codes. In this case we can extend $C$ trying to maintain the dual minimum distance. This will reduce the dimension of $\mathrm{HF}(\mathcal X)$ by one, which then creates the possibility that there are points which are not in $\mathcal X$ but are zeros of the forms in $\mathrm{HF}(\mathcal X)$. Indeed we can try and find extensions of $C$ so that this is the case.

In all of the above we can can construct a $[\![r,r-2k',d]\!]_q$ code from a truncation of length $r$, for some $k' \leqslant k$.

It should be able to extend these methods to make use of the following recent result of Galindo and Hernando \cite[Theorem 1.2]{GH2021}, which is an extension of Theorem~\ref{sigmaortog}.

There is also the possibility to extend these methods to self-othogonal codes, i.e. $C \leqslant C^{\perp}$. This will work well in the case that the characteristic is even, since $\lambda^{q+1}$ is replaced by $\lambda^2$ and all elements in a field of even characterstic have a square root. The role of the Hermitian form is then replaced by a quadratic form.

   Simeon Ball\\
   Departament de Matem\`atiques, \\
Universitat Polit\`ecnica de Catalunya, \\
Carrer Jordi Girona 1-3,\\
08034 Barcelona, Spain \\
   {\tt simeon@ma4.upc.edu} \\

Ricard Vilar\\
   Departament de Matem\`atiques, \\
Universitat Polit\`ecnica de Catalunya, \\
Carrer Jordi Girona 1-3,\\
08034 Barcelona, Spain \\
  {\tt ricard.vilar@upc.edu} \\

\end{document}